\documentclass[conference]{IEEEtran}

\usepackage{amsfonts,amsbsy,bm,amsmath}
\usepackage[nolist]{acronym}
\usepackage{mathrsfs} 
\usepackage{graphicx,cite,amssymb,mathtools,amsthm}
\usepackage{stfloats}
\usepackage{xcolor}
\mathtoolsset{showonlyrefs}
\usepackage{bbm}
\usepackage{tikz}
\usepackage{pgfplots}
\pgfplotsset{compat=1.14}

\usepackage{nicefrac}
\definecolor{myParula01Blue}{RGB}{0,114,189}
\definecolor{myParula02Orange}{RGB}{217,83,25}
\definecolor{myParula03Yellow}{RGB}{237,177,32}
\definecolor{myParula04Purple}{RGB}{126,47,142}
\definecolor{myParula05Green}{RGB}{119,172,48}
\definecolor{myParula06LightBlue}{RGB}{77,190,238}
\definecolor{myParula07Red}{RGB}{162,20,47}

\tikzset{myparula11/.style={color=myParula01Blue,solid,mark=+,mark options={solid}}}
\tikzset{myparula12/.style={color=myParula01Blue,densely dashed,mark=x,mark options={solid}}}
\tikzset{myparula13/.style={color=myParula01Blue,densely dotted,mark=o,mark options={solid}}}
\tikzset{myparula14/.style={color=myParula01Blue,dashdotted,mark=triangle,mark options={solid}}}
\tikzset{myparula15/.style={color=myParula01Blue,dashdotdotted,mark=square,mark options={solid}}}

\tikzset{myparula21/.style={color=myParula02Orange,solid,mark=+,mark options={solid}}}
\tikzset{myparula22/.style={color=myParula02Orange,densely dashed,mark=x,mark options={solid}}}
\tikzset{myparula23/.style={color=myParula02Orange,densely dotted,mark=o,mark options={solid}}}
\tikzset{myparula24/.style={color=myParula02Orange,dashdotted,mark=triangle,mark options={solid}}}
\tikzset{myparula25/.style={color=myParula02Orange,dashdotdotted,mark=square,mark options={solid}}}

\tikzset{myparula31/.style={color=myParula03Yellow,solid,mark=+,mark options={solid}}}
\tikzset{myparula32/.style={color=myParula03Yellow,densely dashed,mark=x,mark options={solid}}}
\tikzset{myparula33/.style={color=myParula03Yellow,densely dotted,mark=o,mark options={solid}}}
\tikzset{myparula34/.style={color=myParula03Yellow,dashdotted,mark=triangle,mark options={solid}}}
\tikzset{myparula35/.style={color=myParula03Yellow,dashdotdotted,mark=square,mark options={solid}}}

\tikzset{myparula41/.style={color=myParula04Purple,solid,mark=+,mark options={solid}}}
\tikzset{myparula42/.style={color=myParula04Purple,densely dashed,mark=x,mark options={solid}}}
\tikzset{myparula43/.style={color=myParula04Purple,densely dotted,mark=o,mark options={solid}}}
\tikzset{myparula44/.style={color=myParula04Purple,dashdotted,mark=triangle,mark options={solid}}}
\tikzset{myparula45/.style={color=myParula04Purple,dashdotdotted,mark=square,mark options={solid}}}

\tikzset{myparula51/.style={color=myParula05Green,solid,mark=+,mark options={solid}}}
\tikzset{myparula52/.style={color=myParula05Green,densely dashed,mark=x,mark options={solid}}}
\tikzset{myparula53/.style={color=myParula05Green,densely dotted,mark=o,mark options={solid}}}
\tikzset{myparula54/.style={color=myParula05Green,dashdotted,mark=triangle,mark options={solid}}}
\tikzset{myparula55/.style={color=myParula05Green,dashdotdotted,mark=square,mark options={solid}}}

\tikzset{myparula61/.style={color=myParula06LightBlue,solid,mark=+,mark options={solid}}}
\tikzset{myparula62/.style={color=myParula06LightBlue,densely dashed,mark=x,mark options={solid}}}
\tikzset{myparula63/.style={color=myParula06LightBlue,densely dotted,mark=o,mark options={solid}}}
\tikzset{myparula64/.style={color=myParula06LightBlue,dashdotted,mark=triangle,mark options={solid}}}
\tikzset{myparula65/.style={color=myParula06LightBlue,dashdotdotted,mark=square,mark options={solid}}}

\tikzset{myparula71/.style={color=myParula07Red,solid,mark=+,mark options={solid}}}
\tikzset{myparula72/.style={color=myParula07Red,densely dashed,mark=x,mark options={solid}}}
\tikzset{myparula73/.style={color=myParula07Red,densely dotted,mark=o,mark options={solid}}}
\tikzset{myparula74/.style={color=myParula07Red,dashdotted,mark=triangle,mark options={solid}}}
\tikzset{myparula75/.style={color=myParula07Red,dashdotdotted,mark=square,mark options={solid}}}

\begin{document}
	\newcommand{\ensemble}{\mathscr{C}}
\newcommand{\code}{\mathcal{C}}
\newcommand{\vecu}{\boldsymbol{u}}
\newcommand{\veci}{\boldsymbol{i}}
\newcommand{\vecf}{\boldsymbol{f}}
\newcommand{\vecv}{\boldsymbol{v}}
\newcommand{\vecp}{\boldsymbol{p}}
\newcommand{\vecx}{\boldsymbol{x}}
\newcommand{\vecone}{\boldsymbol{1}}
\newcommand{\vecuhat}{\hat{\boldsymbol{u}}}
\newcommand{\vecc}{\boldsymbol{c}}
\newcommand{\vecb}{\boldsymbol{b}}
\newcommand{\vecchat}{\hat{\boldsymbol{c}}}
\newcommand{\vecy}{\boldsymbol{y}}
\newcommand{\vecz}{\boldsymbol{z}}
\newcommand{\F}{\boldsymbol{F}}
\newcommand{\B}{\boldsymbol{B}}
\newcommand{\G}{\boldsymbol{G}}
\newcommand{\GK}{\boldsymbol{K}}
\newcommand{\Gsys}{\G_{\mathsf{i},\mathsf{sys}}}
\newcommand{\Gnsys}{\G_{\mathsf{i},\mathsf{nsys}}}
\newcommand{\Per}{\boldsymbol{\Pi}}
\newcommand{\I}{\boldsymbol{I}}
\newcommand{\perP}{\boldsymbol{P}}
\newcommand{\perS}{\boldsymbol{S}}
\newcommand{\GH}[1]{\bm{\mathsf{G}}_{#1}}
\newcommand{\T}[1]{\bm{\mathsf{T}}{(#1)}}

\newcommand{\mi}{\mathrm{I}}
\newcommand{\Prob}{P}
\newcommand{\Z}{\mathrm{Z}}
\newcommand{\SPC}{\mathcal{S}}
\newcommand{\Rep}{\mathcal{R}}
\newcommand{\f}{\mathrm{f}}
\newcommand{\W}{\mathrm{W}}
\newcommand{\A}{\mathrm{A}}
\newcommand{\SC}{\mathrm{SC}}
\newcommand{\Q}{\mathrm{Q}}
\newcommand{\MAP}{\mathrm{MAP}}
\newcommand{\E}{\mathrm{E}}
\newcommand{\U}{\mathrm{U}}
\newcommand{\Ham}{\mathrm{H}}

\newcommand{\Amin}{\mathrm{A}_{\mathrm{min}}}
\newcommand{\dmin}{d}
\newcommand{\erfc}{\mathrm{erfc}}

\newcommand{\de}{\mathrm{d}}

\newcommand{\decoRule}{\rule{\textwidth}{.4pt}}

\newcommand{\oleq}[1]{\overset{\text{(#1)}}{\leq}}
\newcommand{\oeq}[1]{\overset{\text{(#1)}}{=}}
\newcommand{\ogeq}[1]{\overset{\text{(#1)}}{\geq}}
\newcommand{\ogeql}[2]{\overset{#1}{\underset{#2}{\gtreqless}}}
\newcommand{\argmax}[1]{\underset{#1}{\mathrm{arg \, max}}\,}
\newcommand\numeq[1]%
{\stackrel{\scriptscriptstyle(\mkern-1.5mu#1\mkern-1.5mu)}{=}}

\newcommand{\pscd}{\gl{\Prob(\mathcal{E}_{\SCD})}}
\newcommand{\uED}{\gl{\hat{u}}}
\newcommand{\LED}{\gl{L_i^{(\ED)}}}
\newcommand{\inverson}[1]{\gl{\mathbb{I}\left\{#1\right\}}}

\newtheorem{mydef}{Definition}
\newtheorem{prop}{Proposition}
\newtheorem{theorem}{Theorem}

\newtheorem{lemma}{Lemma}
\newtheorem{remark}{Remark}
\newtheorem{example}{Example}
\newtheorem{definition}{Definition}
\newtheorem{corollary}{Corollary}


\definecolor{lightblue}{rgb}{0,.5,1}
\definecolor{normemph}{rgb}{0,.2,0.6}
\definecolor{supremph}{rgb}{0.6,.2,0.1}
\definecolor{lightpurple}{rgb}{.6,.4,1}
\definecolor{gold}{rgb}{.6,.5,0}
\definecolor{orange}{rgb}{1,0.4,0}
\definecolor{hotpink}{rgb}{1,0,0.5}
\definecolor{newcolor2}{rgb}{.5,.3,.5}
\definecolor{newcolor}{rgb}{0,.3,1}
\definecolor{newcolor3}{rgb}{1,0,.35}
\definecolor{darkgreen1}{rgb}{0, .35, 0}
\definecolor{darkgreen}{rgb}{0, .6, 0}
\definecolor{darkred}{rgb}{.75,0,0}
\definecolor{midgray}{rgb}{.8,0.8,0.8}
\definecolor{darkblue}{rgb}{0,.25,0.6}

\definecolor{lightred}{rgb}{1,0.9,0.9}
\definecolor{lightblue}{rgb}{0.9,0,0.0}
\definecolor{lightpurple}{rgb}{.6,.4,1}
\definecolor{gold}{rgb}{.6,.5,0}
\definecolor{orange}{rgb}{1,0.4,0}
\definecolor{hotpink}{rgb}{1,0,0.5}
\definecolor{darkgreen}{rgb}{0, .6, 0}
\definecolor{darkred}{rgb}{.75,0,0}
\definecolor{darkblue}{rgb}{0,0,0.6}

\definecolor{bgblue}{RGB}{245,243,253}
\definecolor{ttblue}{RGB}{91,194,224}

\definecolor{dark_red}{RGB}{150,0,0}
\definecolor{dark_green}{RGB}{0,150,0}
\definecolor{dark_blue}{RGB}{0,0,150}
\definecolor{dark_pink}{RGB}{80,120,90}

\begin{acronym}
	\acro{APP}{a-posteriori probability}
	\acro{AWGN}{additive white Gaussian noise}
	\acro{B-AWGN}{binary input additive white Gaussian noise}
	\acro{B-DMC}{binary-input discrete memoryless channel}
	\acro{B-DMSC}{binary-input discrete memoryless symmetric channel}
	\acro{B-MSC}{binary-input memoryless symmetric channel}
	\acro{BCJR}{Bahl, Cocke, Jelinek, and Raviv}
	\acro{BEC}{binary erasure channel}
	\acro{BER}{bit error rate}
	\acro{BLEP}{block error probability}
	\acro{BLER}{block error rate}
	\acro{BP}{belief propagation}
	\acro{BSC}{binary symmetric channel}
	\acro{CER}{codeword error rate}
	\acro{CN}{check node}
	\acro{CRC}{cyclic redundancy check}
	\acro{DE}{density evolution}
	\acro{eBCH}{extended Bose-Chaudhuri-Hocquengham}
	\acro{FER}{frame error rate}
	\acro{GA}{Gaussian approximation}
	\acro{i.i.d.}{independent and identically distributed}
	\acro{IO-WE}{input-output weight enumerator}
	\acro{IR-WE}{input-redundancy weight enumerator}
	\acro{IO-WEF}{input-output weight enumerating function}
	\acro{IR-WEF}{input-redundancy weight enumerating function}
	\acro{JIO-WE}{joint IO-WE}
	\acro{JIR-WE}{joint IR-WE}
	\acro{JWE}{joint WE}
	\acro{LDPC}{low-density parity-check}
	\acro{LHS}{left-hand side}
	\acro{LLR}{log-likelihood ratio}
	\acro{MAP}{maximum a-posteriori}
	\acro{MC}{metaconverse}
	\acro{ML}{maximum likelihood}
	\acro{NA}{normal approximation}
	\acro{PC}{product code}
	\acro{pdf}{probability density function}
	\acro{RCB}{random coding bound}
	\acro{RCUB}{random coding union bound}
	\acro{RM}{Reed--Muller}
	\acro{r-RM}{random RM}
	\acro{RHS}{right-hand side}
	\acro{RV}{random variable}
	\acro{SPC}{single parity-check}
	\acro{SC}{successive cancellation}
	\acro{SCC}{super component codes}
	\acro{SCL}{successive cancellation list}
	\acro{SISO}{soft-input soft-output}
	\acro{SNR}{signal-to-noise ratio}
	\acro{UB}{union bound}
	\acro{TUB}{truncated union bound}
	\acro{VN}{variable node}
	\acro{WE}{weight enumerator}
	\acro{WEF}{weight enumerating function}
\end{acronym}
	\title{Successive Cancellation Inactivation Decoding for Modif\hspace{0.00001mm}ied Reed-Muller and eBCH Codes}

	\author{\IEEEauthorblockN{Mustafa Cemil Co\c{s}kun\IEEEauthorrefmark{1}\IEEEauthorrefmark{2}\IEEEauthorrefmark{3}, Joachim Neu\IEEEauthorrefmark{7} and Henry D. Pfister\IEEEauthorrefmark{3}}
	
	\IEEEauthorblockA{\IEEEauthorrefmark{1}German Aerospace Center (DLR)\, \IEEEauthorrefmark{2}Technical University of Munich (TUM)\, \IEEEauthorrefmark{3}Duke University\, \IEEEauthorrefmark{7}Stanford University \\
	Email: mustafa.coskun@tum.de, jneu@stanford.edu, henry.pfister@duke.edu
	}
}

	\maketitle

	\begin{abstract}
		A successive cancellation (SC) decoder with inactivations is proposed as an efficient implementation of SC list (SCL) decoding over the binary erasure channel. The proposed decoder assigns a dummy variable to an information bit whenever it is erased during SC decoding and continues with decoding. Inactivated bits are resolved using information gathered from decoding frozen bits. This decoder leverages the structure of the Hadamard matrix, but can be applied to any linear code by representing it as a polar code with dynamic frozen bits. SCL decoders are partially characterized using density evolution to compute the average number of inactivations required to achieve the maximum a-posteriori decoding performance. The proposed measure quantifies the performance vs. complexity trade-off and provides new insight into dynamics of the number of paths in SCL decoding.
		The technique is applied to analyze Reed-Muller (RM) codes with dynamic frozen bits. It is shown that these modified RM codes perform close to extended BCH codes.
	\end{abstract}

	\section{Introduction}\label{sec:intro}

Since their introduction in \cite{muller}, various decoding algorithms for \ac{RM} codes have been proposed to achieve performance close to \ac{MAP} with reduced complexity (see, e.g., \cite{reed,Dumer04,Dumer06,Dumer:List06,Santi18,mengke19,Ye19,Ivanov19}).
Recently, it has been shown that \ac{RM} codes can achieve capacity on the \ac{BEC} under \ac{MAP} decoding \cite{kudekar17}.

Polar codes are the first provably capacity-achieving codes with low encoding and decoding complexity for arbitrary symmetric \acp{B-DMC}\cite{arikan2009channel}.
Using \ac{SCL} decoding \cite{scld}, the addition of a high-rate outer code can make them very competitive in the short- to moderate-length regime (i.e., from $128$ to $1024$ bits) \cite{Coskun18:Survey}.
Since \ac{RM} codes are closely related to polar codes \cite{arikan2009channel} but outperform them under \ac{MAP} decoding \cite{Mondelli14}, some decoders proposed for polar codes have been also used for \ac{RM} codes, e.g., \cite{Ivanov19}. As the complexity of the \ac{SCL} decoder tends to be large when used to decode \ac{RM} codes, hybrid designs \cite{Mondelli14,RMpolar14,FabregasSiegel17} have been considered to trade performance for decoding complexity. In addition, the authors of \cite{trifonov16} showed how any linear code could be viewed as a polar code with dynamic frozen bits. They also constructed subcodes of \ac{eBCH} codes and decoded them using \ac{SCL} decoding by representing them as polar codes with dynamic frozen bits. These codes, dubbed \ac{eBCH}-polar subcodes, allow one to trade complexity for performance. Different design algorithms of polar code variants for a given list size are provided by \cite{Yuan19,Rowshan19}.

In this work, \emph{\ac{SC} inactivation} decoding is proposed. It follows the same message passing schedule as the \ac{SCL} decoder. However, whenever an information bit decodes as erased, it is replaced by a dummy variable, i.e., it is inactivated, and the decision on it is postponed to the end of decoding process. The inactivated bits are resolved using information gathered from decoding frozen bits. This decoder is subsequently extended to solve for the inactivated bits during the decoding process.

Similar decoders have been proposed in the past. In particular, they have been studied for iterative \ac{BP} decoding of \ac{LDPC} \cite{ldpc} and raptor codes \cite{Shokrollahi06}. They are known to overcome high error floors due to stopping sets and to provide \ac{MAP} performance with a lower complexity than standard Gaussian elimination \cite{Pishro04,inactivation05,Measson08,Fran17}. A \ac{BP} decoder with inactivations was proposed for polar codes in \cite{Eslami10}, yielding an improved bit-error rate, but it appears to use a different decoding schedule. The authors of \cite{Measson08} proposed and analyzed the Maxwell decoder for \ac{LDPC} codes, which guesses a value for an erased bit whenever the \ac{BP} decoder is stuck. Their results demonstrate a fundamental relationship between \ac{BP} and \ac{MAP} decoding based on guessing. Inspired by that approach, we analyze the \ac{SC} inactivation decoder to quantify the complexity required to achieve \ac{MAP} performance. Based on the dynamics of the unresolved inactivations during the decoding process, new insights are provided to understand performance vs. complexity trade-off.

Numerical results are provided not only for standard polar and \ac{RM} codes but also for their variants with dynamic frozen bits. In particular, \ac{RM} codes with dynamic frozen bits perform close to \ac{eBCH} codes, which are known to be one of the best performing codes for short block-lengths\cite{Coskun18:Survey}.
	\section{Preliminaries}\label{sec:prelim}
In the following, $ x_a^b $ denotes the vector $ (x_a, x_{a+1}, \dots, x_b) $. If $ b < a $, it is void. Given a vector $x_1^n$ and a set $\mathcal{A}\subset\{1,\dots,n\}$, we write $x_{\mathcal{A}}$ for the subvector $(x_i:i\in\mathcal{A})$. The notation $x_{\mathcal{A}}\cdot y_{\mathcal{A}}$ is used for dot product of two binary vectors. The length-$n$ all-zero vector is denoted as $0^n$. We use capital letters for \acp{RV} and lower case letters for their realizations. We denote a \ac{B-DMC} by $ W : \mathcal{X} \rightarrow \mathcal{Y} $, with input alphabet $ \mathcal{X} = \{0,1\} $, output alphabet $ \mathcal{Y} $, and transition probabilities $ W(y|x) $ for $ x \in \mathcal{X} $ and $ y \in \mathcal{Y} $. The transition probabilities of $n$ independent uses of the same channel are denoted as $ W^n(y_1^n|x_1^n) = \prod_{i=1}^{n} W(y_i|x_i) $. We write \ac{BEC}($ \epsilon $) for the \ac{BEC} with erasure probability $ \epsilon $. The output alphabet of the \ac{BEC} is $\mathcal{Y} = \{0,1,\texttt{e}\}$, where $\texttt{e}$ denotes an erasure. We consider the \ac{BEC} in this work unless otherwise stated. The indicator function $\mathbbm{1}\{\mathrm{P}\}$ equals $1$ if the proposition $\mathrm{P}$ is true and $0$ otherwise. We use capital bold letters for matrices. For example, $\B_n$ denotes the $n\times n$ \emph{bit reversal} matrix \cite{arikan2009channel} and $\GK_{2}$ denotes the $2\times 2$ Hadamard matrix, i.e.,
\[
\GK_{2}\triangleq\begin{bmatrix}
1 & 0 \\ 
1 & 1
\end{bmatrix}.
\]

\subsection{Polar and Reed-Muller Codes}

\label{sec:polar_RM}

Consider the matrix $\G_{n} = \B_n\GK_{2}^{\otimes m}$, where $n\triangleq 2^m$ and $\GK_{2}^{\otimes m}$ is the $m$-fold Kronecker product of $\GK_{2}$. Both polar and \ac{RM} codes are generated by suitable row choices from $\G_{n}$\cite{stolte2002rekursive,arikan2009channel}.
Using $\G_{n}$, the transition probability from the input $u_1^n$ to the output $y_1^n$ is $W_n(y_1^n|u_1^n) \triangleq W^n(y_1^n|u_1^n\G_{n})$. Transition probabilities of the $i$-th \emph{bit-channel}, an artificial channel with the input $u_i$ and the output $(y_1^n,u_1^{i-1})$, are defined by
\begin{equation}
	W_n^{(i)}(y_1^n,u_1^{i-1}|u_i) \triangleq \sum_{u_{i+1}^n\in\mathcal{X}^{n-i}}\frac{1}{2^{n-1}}W_n(y_1^n|u_1^n).
\end{equation}
The code itself is defined by the set of information indices, $\mathcal{A}$. For example, an $(n,k)$ polar code is designed by finding the $k$ most reliable bit-channels with indices $i\in\{1,2,\dots,n\}$ under the assumption that $U_i$, $i\in\{1,2,\dots,n\}$, are \ac{i.i.d.}
uniform \acp{RV}. For a particular channel parameter, these indices can be found using density evolution \cite{arikan2009channel,Mori:2009CL}. For an $r$-th order \ac{RM} code of length $n$ and dimension $k = \sum_{i=0}^{r}\binom{m}{i}$, where $0\leq r\leq m$, the set $\mathcal{A}$ consists of the indices, $i\in\{1,2,\dots,n\}$, corresponding to the rows of $\G_{n}$ with the Hamming weight at least equal to $2^{m-r}$. In both cases, encoding is performed via $c_1^n = u_1^n\G_{n}$, with $u_i=0$ for the indices $i\in\mathcal{A}^c$ of \emph{frozen bits}. The remaining $k$ positions, $u_i$ for $i\in\mathcal{A}$, are allocated for information bits.

\subsection{Representing a Linear Code as a Variant of a Polar Code}
\label{sec:polar_ebch}
Let $\code$ be an $(n,k)$ code with a full-rank parity-check matrix $\boldsymbol{H}$. Consider a vector, $c_1^n$, defined by $c_1^n \triangleq u_1^n\G_{n}$. One can impose linear constraints on $u_1^n$, e.g., $u_1^n\boldsymbol{V}^T = 0^{n-k}$, such that $c_1^n$ is a codeword of $\code$. To this end, we write $c_1^n\boldsymbol{H}^T = u_1^n\G_{n}\boldsymbol{H}^T = 0^{n-k}$. Thus, choosing $\boldsymbol{V} = \boldsymbol{H}\G_n^T$ defines constraints on input vector, $u_1^n$, such that the code $\code$ is represented as a polar code with constraints on the frozen bits\cite{trifonov16}. Using this approach, a bit $u_i$ is called frozen if its value is always $0$ or determined solely by the bits $u_1^{i-1}$. Those which are not always $0$ are called dynamic frozen bits. A systematic way to determine the frozen indices and the constraints on them is to convert $\boldsymbol{V}$ into $\boldsymbol{V}'$ via elementary row operations, where each column of $\boldsymbol{V}'$ has the last non-zero entry in a distinct row $i$, $1\leq i\leq n-k$. Then, the row index $i$ of the non-zero entry in the columns of $\boldsymbol{V}'$ with a single $1$ means that $u_i$ is a static frozen bit, i.e., $u_i=0$. However, the bit $u_i$ is a dynamic frozen bit if a column has multiple $1$'s and $i$ is the row index of its last non-zero entry.
A dynamic frozen bit can be a linear combination of multiple information bits and this is reflected by a column in $\boldsymbol{V}'$ with the Hamming weight greater than $2$. For a given construction, we say the dynamic frozen bits are defined by $t$ information bits if $t$ is the cardinality of the subset of $\mathcal{A}$, consisting of all the indices of information bits used to define dynamic frozen bits.

In \cite{trifonov16}, the idea of dynamic frozen bits was used to define \ac{eBCH}-polar subcodes. The construction of an $(n,k)$ \ac{eBCH}-polar subcode starts by choosing $\code$ to be an $(n,k')$ \ac{eBCH} code with $k < k'$. Then, the frozen indices (for $\code$) are found. Finally, an additional set of $k'-k$ bits are frozen to $0$.

\subsection{Successive Cancellation and SC List Decoding over \ac{BEC}}

Upon observing the channel output $y_1^n$, the \ac{SC} decoding estimates the bit $u_i$ successively from $i=1$ to $i=n$ as
\begin{equation}
	\hat{u}_i = \left\{\begin{array}{lll}
	u_i &\text{if} \,\,\, i\in\mathcal{A}^c \\
	f_i(y_1^n,\hat{u}_1^{i-1})& \text{otherwise,}
	\end{array}
	\right.
	\label{eq:decision}
\end{equation}
by using the previously estimated bits $\hat{u}_1^{i-1}$ in the function
\begin{equation}
\hspace*{-2mm}f_i(y_1^n,\hat{u}_1^{i-1}) \triangleq \left\{\begin{array}{lll}
0& \text{if}& \Prob_{U_i|Y_1^n,U_1^{i-1}}(0|y_1^n,\hat{u}_1^{i-1}) = 1 \\
\texttt{e}& \text{if}& \Prob_{U_i|Y_1^n,U_1^{i-1}}(0|y_1^n,\hat{u}_1^{i-1}) = \frac{1}{2} \\
1& \text{if}& \Prob_{U_i|Y_1^n,U_1^{i-1}}(0|y_1^n,\hat{u}_1^{i-1}) = 0
\end{array}
\right.
\label{eq:decision_fnc}
\end{equation}
where the probabilities $\Prob_{U_i|Y_1^n,U_1^{i-1}}(0|y_1^n,\hat{u}_1^{i-1})$ are computed recursively under the assumption that $U_j$, $i<j\leq n$, are \ac{i.i.d.} uniform random bits\cite{arikan2009channel}. The inputs to the algorithm are $\Prob_{C_i|Y_i}(0|y_i)$, $i=1,\dots,n$. For standard \ac{SC} decoding, the process aborts with a frame error if $\hat{u}_i=\texttt{e}$ for any $i$.

We recall the genie-aided \ac{SC} decoder where, at each decoding stage $i$, the decoder is provided with the true prior bits $u_1^{i-1}$ by a genie\cite{arikan2009channel}. Let $P^{\mathrm{(SC)}}$ and $P^{\mathrm{(GA)}}$ denote the block error probabilities of the \ac{SC} and genie-aided \ac{SC} decoders, respectively. Then, we have the following relation \cite{Mori:2009SIT,Sasoglu12}.
\begin{prop}\label{prop:real_vs_genie}
	For any fixed $\mathcal{A}$, $P^{\mathrm{(SC)}} = P^{\mathrm{(GA)}}$.
\end{prop}
The proof is provided in the appendix for completeness and is valid for any \ac{B-DMC}.
Let $\epsilon_i\triangleq\Prob_{U_i|Y_1^n,U_1^{i-1}}(\texttt{e}|y_1^n,u_1^{i-1})$. Then, $P^{\mathrm{(SC)}} = P^{\mathrm{(GA)}}\leq\sum_{i\in\mathcal{A}}\epsilon_i$ due to Proposition \ref{prop:real_vs_genie}. Thus, the design of polar codes minimizes this upper bound on the block error probability for a given channel erasure rate $\epsilon$.

In \ac{SCL} decoding\cite{scld}, several instances of an \ac{SC} decoder are run in parallel, each having a different hypothesis on the previous estimates $\hat{u}_1^{i-1}$ at a decoding stage $i$. Each hypothesis $\hat{u}_1^{i-1}$ is referred to as a decoding path. Over the \ac{BEC}, if \eqref{eq:decision} provides an erasure for an information bit $u_i$ with $i\in\mathcal{A}$, then a path is duplicated to take into account both options, namely $\hat{u}_i = 0$ and $\hat{u}_i = 1$\cite{Mondelli14,Mondelli15}. Also, whenever a frozen bit is encountered, the number of active paths is halved if its decoded value contradicts its known value because this reveals that half of the paths are invalid \cite{Neu2019}. Note that the complexity of \ac{SCL} decoding is constrained by imposing a maximum list size $L$. We declare an error if, at any stage during the decoding process, the number of active paths exceeds $L$ or if more than one path is active at the end of the process.
	\section{SC Inactivation Decoding}
\label{sec:SC_guessing}
Assume $u_1^n$ is encoded and transmitted over the \ac{BEC}$(\epsilon)$, providing the channel output $y_1^n$. Consider the case where \eqref{eq:decision} provides an erasure, i.e., $\hat{u}_i = \texttt{e}$. Instead of duplicating the path as in \ac{SCL} decoding, the \ac{SC} inactivation decoder introduces a dummy variable $\tilde{u}_i$ and stores the decision as $\hat{u}_i = \tilde{u}_i$, called an inactivation event. It continues decoding with the next stages using the same schedule as for \ac{SC} decoding. Then, \eqref{eq:decision_fnc} is allowed to be a function of the previous inactivated variable. For example, assuming there is no other inactivation for the information bits in between, it can output either an erasure or a linear combination of $\tilde{u}_0 \triangleq 1$ and the previous variable $\tilde{u}_i$, i.e., $a_0\oplus a_i\tilde{u}_i$ with $a_{\{0,i\}}\in\{0,1\}^2$, for all bits $u_j$ with $j>i$. We have to separate the cases where (i) $u_j$ is an information bit and (ii) $u_j$ is a frozen bit. In case (i), if the function \eqref{eq:decision_fnc} outputs an erasure, the decoder inactivates another bit, namely $\hat{u}_j = \tilde{u}_j$. Otherwise, it continues with the decoding of the next bit by knowing that $\hat{u}_j = a_0\oplus a_i\tilde{u}_i$. In case (ii), if the decoder outputs an erasure or has a trivial combination, i.e, $a_{\{0,i\}}=0^2$, it sets $\hat{u}_j = 0$ and continues with the next bit. However, if it outputs a combination where $a_i=1$, then it learns the value of the previously inactivated bit as $\tilde{u}_i = a_0$. The \ac{SC} inactivation decoder stores the equation separately and keeps decoding with $\hat{u}_j=a_0\oplus a_i\tilde{u}_i$. 

In general, the decoder can have $g$ inactivations for the information bits $u_{\mathcal{G}\backslash\{0\}}$ with $\mathcal{G}\triangleq\{0,i_1,i_2,\dots,i_g\}$, $0<i_1<i_2<\dots<i_g<i$ before decoding $u_i$, i.e., $\hat{u}_{\mathcal{G}\backslash\{0\}} = \tilde{u}_{\mathcal{G}\backslash\{0\}}$. 
For some binary vector $a_\mathcal{G}$, the function $f_i$ is
\vspace{-1mm}
\begin{equation}
f_i(y_1^n,\hat{u}_1^{i-1}) \hspace{-1mm}\triangleq \hspace{-1mm}\left\{\begin{array}{lll}
\hspace{-2mm}a_{\mathcal{G}}\cdot \tilde{u}_{\mathcal{G}} & \text{if}\,\, \Prob_{U_i|Y_1^n,U_1^{i-1}}(a_{\mathcal{G}} \cdot \tilde{u}_{\mathcal{G}}|y_1^n,\hat{u}_1^{i-1}) = 1 \\
\texttt{e} &\text{otherwise.}
\end{array}
\right.
\label{eq:decision_fnc_guess}
\end{equation}
Assume that the decoder inactivates $g$ bits in total during a decoding attempt. Then, the final step of \ac{SC} inactivation decoding is to solve a system of linear equations in $g$ unknowns. This will have a unique solution only if the equations obtained from frozen bits have rank $g$. This algorithm is equivalent to an \ac{SCL} decoder over the \ac{BEC} with unbounded list size \cite[Appendix A]{Neu2019}, thus, it implements \ac{MAP} decoding.

We extend the inactivation decoder to include path pruning like \ac{SCL} decoding. The decoder's operation is unchanged whenever an information bit is encountered, i.e., in case (i) above. In case (ii), if \eqref{eq:decision_fnc_guess} does not deliver an erasure, it provides the equation $a_{\mathcal{G}}\cdot\tilde{u}_\mathcal{G} = 0$.\footnote{If it is a dynamic frozen bit, the right-hand side of the equation is the linear combination defining it (see Sec. \ref{sec:polar_ebch}). For simplicity, assume it is not.} If $a_\mathcal{G}$ has a non-zero term, the equation is solved for $\tilde{u}_{i_j}$ as $ \tilde{u}_{i_j} = a_{\mathcal{G}\backslash \{i_j\}}\cdot \tilde{u}_{\mathcal{G}\backslash \{i_j\}}$, $i_j = \max\{i\in\mathcal{G}:i\neq0\}$, and stored. This is called a \emph{consolidation} event. The decoder continues with $\hat{u}_j=0$. We declare an error if there remains any unresolved $\tilde{u}_{i}$ at the end.
Analysis of this decoder provides insights into dynamics of the number of paths in \ac{SCL} decoding for the \ac{BEC}.
	\section{Number of Inactivations for MAP Decoding}
Let $\hat{u}_1^{i-1}$ denote the decoding output of \ac{SC} inactivation decoding with possible inactivations before estimating $u_i$.
\begin{lemma}
	$f_i(y_1^n,u_1^{i-1}) = \texttt{e}$ if and only if $f_i(y_1^n,\hat{u}_1^{i-1}) = \texttt{e}$.
	\label{lem:guessing_vs_genie}
\end{lemma}
\begin{proof}
	The case where $\hat{u}_1^{i-1}$ does not contain any inactivation, i.e., $\hat{u}_1^{i-1}=u_1^{i-1}$, is trivial.
	
	Thus, we assume that the decoder inactivated some information bits, i.e., $\hat{u}_j = \tilde{u}_j$ for some $j$, $1\leq j<i$.
	
	Now, suppose that (a) $f_i(y_1^n,u_1^{i-1}) = \texttt{e}$, meaning $\Prob_{U_i|Y_1^n,U_1^{i-1}}(u_i|y_1^n,u_1^{i-1}) = \nicefrac{1}{2}$ for both $u_i\in\{0,1\}$, and that (b) we have a vector $a_\mathcal{G}$ such that $\Prob_{U_i|Y_1^n,U_1^{i-1}}(a_\mathcal{G}\cdot \tilde{u}_\mathcal{G}|y_1^n,\hat{u}_1^{i-1}) = 1$; thus, $f_i(y_1^n,\hat{u}_1^{i-1}) \neq \texttt{e}$. But, (b) implies $\Prob_{U_i|Y_1^n,U_1^{i-1}}(a_\mathcal{G}\cdot u_\mathcal{G}|y_1^n,u_1^{i-1}) = 1$ by replacing inactivated bits with their values and having $u_0 = \tilde{u}_0$ (equivalently, $u_0 = 1$). This contradicts (a).
	
	Now, consider the other direction, i.e., suppose that (c) $f_i(y_1^n,u_1^{i-1}) = u_i$ and (d) $f_i(y_1^n,\hat{u}_1^{i-1}) = \texttt{e}$. Then, (c) implies that $\Prob_{U_i|Y_1^n,U_1^{i-1}}(u_i|y_1^n,u_1^{i-1}) = 1$ for some $u_i\in\{0,1\}$. Then, there exists a vector $a_\mathcal{G}$ for which we have
	\vspace{-1mm}
	\begin{equation}
		 \Prob_{U_i|Y_1^n,U_1^{i-1}}(a_\mathcal{G}\cdot\tilde{u}_\mathcal{G}|y_1^n,\hat{u}_1^{i-1})\Big\rvert_{a_\mathcal{G}\cdot\tilde{u}_\mathcal{G} = u_i, \hat{u}_1^{i-1} = u_1^{i-1}} = 1
		 \vspace{-1mm}
	\end{equation}
	and this contradicts (d).
\end{proof}
\vspace*{-2mm}
Note that density evolution is able to compute the probabilities $\epsilon_i$ exactly, i.e., the erasure probabilities of the genie-aided \ac{SC} decoder\cite{arikan2009channel}.
\vspace*{-2mm}
\begin{lemma}
	Let $b_i$ denote the probability of having an inactivation for $u_i$ in the \ac{SC} inactivation decoder. Then,
	\[b_i = \left\{\begin{array}{lll}
	0 &\text{if} \,\,\, i\in\mathcal{A}^{c} \\
	\epsilon_i & \text{otherwise.}
	\end{array}
	\right.
	\]
	\label{lem:guessing_prob}
\end{lemma}
\vspace*{-6mm}
\begin{proof}
	Let $b_i'$ denote $\Pr[f_i(y_1^n,\hat{u}_1^{i-1}) = \texttt{e}]$. Then, we have
	$b_i' \numeq{\text{a}} \E[\mathbbm{1}\{f_i(y_1^n,\hat{u}_1^{i-1}) = \texttt{e}\}] \numeq{\text{b}} \E[\mathbbm{1}\{f_i(y_1^n,u_1^{i-1}) = \texttt{e}\}] \numeq{\text{c}} \epsilon_i$ where ($\text{a}$) and ($\text{c}$) follows from the definition of probability and ($\text{b}$) from Lemma \ref{lem:guessing_vs_genie}. The result follows from \eqref{eq:decision}.
\end{proof}
\vspace*{-3mm}
\begin{corollary}
	Let $G$ be a \ac{RV} equal to the total number of inactivations made by the decoder during a decoding attempt. Then, $\E[G] = \sum_{i\in\mathcal{A}}\epsilon_i$.
	\label{cor:guessing}
\end{corollary}
\vspace*{-3mm}
\begin{proof}
	Follows from the expectation of inactivation indicator events and Lemma \ref{lem:guessing_prob}. 
\end{proof}
\vspace*{-2mm}
Corollary \ref{cor:guessing} describes the average number of inactivations required with a code defined by $\mathcal{A}$ for \ac{MAP} performance when the transmission is over the \ac{BEC}$(\epsilon)$, providing the expected number of unknowns for the resulting linear system.\footnote{This relation was first observed in \cite[Appendix A]{Neu2019} for \ac{SCL} decoding with unbounded list size, where an inactivation event is replaced by a branching event as an \ac{SCL} decoder \emph{branches} paths if it encounters an erasure.}

For the \ac{SC} inactivation decoder with consolidations, let $G_i$ be a \ac{RV} denoting the number of unresolved inactivations after decoding $u_i$. Its expectation $\E[G_i]$ is important for analyzing the error probability of an \ac{SC} inactivation decoder with a fixed maximum number of inactivated bits $g_{\mathrm{max}}$ as well as that of an \ac{SCL} decoder with a list size $L$  set to $L=2^{g_{\mathrm{max}}}$.
In Section \ref{sec:numerical}, numerical estimates of $\E[G_i]$ are shown.

\begin{remark}
	The performance improvement under \ac{MAP} decoding when interpolating from polar to \ac{RM} codes is driven by the weight spectrum improvement, e.g., the minimum distance increases\cite{RMpolar14,Mondelli14}. This comes at the cost of a higher MAP decoding complexity. The quantity $\E[G]$ is obtained from analyzing the inactivation decoder and this allows us to quantify this complexity increase. Another way to improve the distance spectrum is to embed dynamic frozen bits \cite{trifonov16}. Although the number of inactivated bits remains unaffected by the use of dynamic frozen bits, they do add extra complexity to the decoder. The additional complexity is related to the total number of information bits used to define dynamic frozen bits because it  affects the sparsity of the binary vector operations.
\end{remark}
	\section{Numerical Results}
\label{sec:numerical}
In this section, we consider (i) polar codes, (ii) \ac{RM} codes, (iii) \ac{eBCH}-polar subcodes (with $7$ dynamic frozen bits) \cite{Yuan19} as well as (iv) \ac{RM} codes with dynamic frozen bits. Numerical results are provided for rate $R = \nicefrac{1}{2}$ codes of length $n\in\{128,512\}$ using a \ac{MAP} decoding implemented via the \ac{SC} inactivation decoder. Note that the polar codes are designed for the erasure probability $\epsilon=0.4$ via density evolution \cite{arikan2009channel}. The Singleton bound (SB)\cite{Singleton64}, a lower bound on the block error probability of any binary linear code, and the Berlekamp random coding bound (BRCB) \cite{Berlekamp80}, a tight upper bound on the average block error probability of the linear code ensemble defined via parity-check matrices, are provided as benchmark.

In Fig. \ref{fig:performance}, the \acp{BLER} are shown.
For any length $n$, \ac{RM} codes outperform polar codes. The \ac{eBCH} code performs very close to an instance from the $(128,64)$ random code ensemble. It is constructed with the idea explained in Sec. \ref{sec:polar_ebch} by allocating the first $k$ positions for information bits and having all frozen bits as dynamic, which are set to random linear combinations of all information bits. The \ac{BLER} of the \ac{eBCH}-polar subcode is also provided as a reference and it performs slightly better than the \ac{RM} code.
	
In Fig. \ref{fig:complexity}, the expected numbers of inactivations $\E[G]$ from Corollary \ref{cor:guessing}  are provided together with the results obtained from simulations for $n=128$, demonstrating that the analysis is exact. In addition to larger number of inactivations, the \ac{eBCH} code has $35$ dynamic frozen bits, which are defined by in total $35$ information bits. Therefore, the average decoding complexity for \ac{eBCH} code is much higher than for the others. For this blocklength, surprisingly, $\E[G_{\mathrm{RM}}]$ is close to $\E[G_{\mathrm{eBCH-pol}}]$, where the \ac{eBCH}-polar subcode has an additional complexity due to dynamic frozen bits.

In Fig. \ref{fig:performance_dynamic}, the \acp{BLER} for two variants of \ac{RM} codes with dynamic frozen bits are provided. In d-RM codes, all frozen bits are set to random linear combinations of the previous information bit(s). For $n=128$, the d-RM code's performance is close to that of the \ac{eBCH} code. The second variant, $7$d-RM code, is designed for a lower decoding complexity compared to the d-RM code by declaring all but the last $7$ frozen bits as static. The dynamic frozen bits are set to random linear combinations of first $10$ information bits since they are more likely to be erased. This code performs within an erasure probability gap of $0.04$ from the \ac{eBCH} code at a \ac{BLER} $\approx 10^{-6}$. For $n=512$, the flattening in the curve of the \ac{RM} code at \ac{BLER} $\approx 10^{-5}$ is avoided by the d-\ac{RM}, performing close to the SB down to a \ac{BLER} $\approx 10^{-7}$.
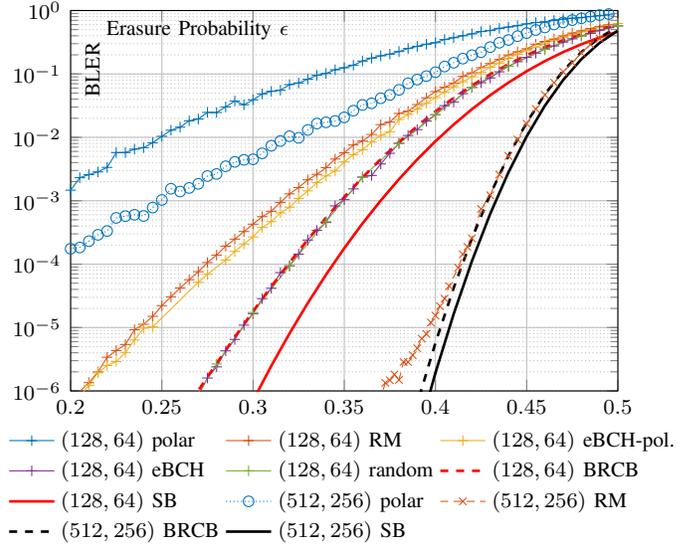
\begin{figure}    
	\centering
	\begin{tikzpicture}
\begin{semilogyaxis}[
	mark options={solid,scale=0.75},
    width=1*\linewidth,
    height=0.75*\columnwidth,
    title style={font=\footnotesize,align=center},
    legend cell align=left,
    legend style={font=\footnotesize},
    legend columns=3,
    legend style={at={(0.5,-0.075)},anchor=north,draw=none,/tikz/every even column/.append style={column sep=0mm},cells={align=left}},
    ylabel near ticks,
    xlabel near ticks,
	xmin=0.2,
	xmax=0.5,
	ymin=1e-6,
	ymax=1,
	xlabel={\textcolor{black}{Erasure Probability $\epsilon$}},
	ylabel={\textcolor{black}{BLER}},
	x label style={at={(axis description cs:0.0,0.9)},anchor=south west,yshift=0em,xshift=1em},
	y label style={at={(axis description cs:0.01,0.75)},anchor=north west,yshift=0em,xshift=0em},
	grid=both,
	major grid style={solid,draw=gray!50},
	minor grid style={densely dotted,draw=gray!50},
	label style={font=\footnotesize},
	tick label style={font=\footnotesize},
]

\addplot[myparula11] table[x=erasure,y=P] {figures/128_64_polar_bec_designE=04.txt};
\addlegendentry{$(128,64)$ polar}
\addplot[myparula21] table[x=erasure,y=P] {figures/128_64_RM_bec.txt};
\addlegendentry{$(128,64)$ RM}
\addplot[myparula31] table[x=erasure,y=P] {figures/128_64_eBCHpeihong_bec.txt};
\addlegendentry{$(128,64)$ eBCH-pol.}
\addplot[myparula41] table[x=erasure,y=P] {figures/128_64_ebch_bec.txt};
\addlegendentry{$(128,64)$ eBCH}
\addplot[myparula51] table[x=erasure,y=P] {figures/128_64_random_bec.txt};
\addlegendentry{$(128,64)$ random}
\addplot[color=red,line width=1pt,dashed] table[x=erasure,y=P] {figures/128_64_berlekamp_bec.txt};
\addlegendentry{$(128,64)$ BRCB}
\addplot[color=red,line width=1pt,solid] table[x=erasure,y=P] {figures/128_64_singleton_bec.txt};
\addlegendentry{$(128,64)$ SB}

\addplot[myparula13] table[x=erasure,y=P] {figures/512_256_polar_bec_designE=04.txt};
\addlegendentry{$(512,256)$ polar}
\addplot[myparula22] table[x=erasure,y=P] {figures/512_256_RM_bec.txt};
\addlegendentry{$(512,256)$ RM}
\addplot[color=black,line width = 1pt,dashed] table[x=erasure,y=P] {figures/512_256_berlekamp_bec.txt};
\addlegendentry{$(512,256)$ BRCB}
\addplot[color=black,line width = 1pt,solid] table[x=erasure,y=P] {figures/512_256_singleton_bec.txt};
\addlegendentry{$(512,256)$ SB}

\end{semilogyaxis}
\end{tikzpicture}
	\vspace*{-8mm}
	\caption{BLER vs. $\epsilon$ for different codes.}\label{fig:performance}
\end{figure}
\begin{figure}    
	\centering
	\begin{tikzpicture}
\begin{axis}[
	mark options={solid,scale=0.75},
    width=1*\linewidth,
    height=0.65*\columnwidth,
    title style={font=\footnotesize,align=center},
    legend cell align=left,
    legend style={font=\footnotesize},
    legend columns=2,
    legend style={at={(0.5,-0.1)},anchor=north,draw=none,/tikz/every even column/.append style={column sep=1mm},cells={align=left}},
    ylabel near ticks,
    xlabel near ticks,
	xmin=0.1,
	xmax=0.5,
	ymin=0,
	ymax=15,
	xlabel={\textcolor{black}{Erasure Probability $\epsilon$}},
	ylabel={\textcolor{black}{Inactivations $\E[G]$}},
	x label style={at={(axis description cs:0.01,0.875)},anchor=south west,yshift=0em,xshift=1em},
	y label style={at={(axis description cs:0.0,0.42)},anchor=north west,yshift=0em,xshift=0em},
	grid=both,
	major grid style={solid,draw=gray!50},
	minor grid style={densely dotted,draw=gray!50},
	label style={font=\footnotesize},
	tick label style={font=\footnotesize},
]

\addplot[color=myParula04Purple,line width = 1pt,solid] table[x=erasure,y=MI_loss] {figures/128_64_eBCH_bec_MI_loss.txt};
\addlegendentry{$\E[G_{\mathrm{eBCH}}]$}

\addplot[color=myParula02Orange,line width = 1pt,solid] table[x=erasure,y=MI_loss] {figures/128_64_RM_bec_MI_loss.txt};
\addlegendentry{$\E[G_{\mathrm{RM}}]$}

\addplot[color=myParula03Yellow,line width = 1pt,solid] table[x=erasure,y=MI_loss] {figures/128_64_eBCHpeihong_bec_MI_loss.txt};
\addlegendentry{$\E[G_{\mathrm{eBCH-pol}}]$}

\addplot[color=myParula01Blue,line width = 1pt,solid] table[x=erasure,y=MI_loss] {figures/128_64_polar_bec_MI_loss_designE=04.txt};
\addlegendentry{$\E[G_{\mathrm{polar}}]$}

\addplot[color=myParula04Purple,only marks,mark=x] table[x=erasure,y=branching] {figures/128_64_eBCH_bec_branching_all.txt};

\addplot[color=myParula02Orange,only marks,mark=x] table[x=erasure,y=branching] {figures/128_64_RM_bec_branching_all.txt};

\addplot[color=myParula03Yellow,only marks,mark=x] table[x=erasure,y=branching] {figures/128_64_eBCHpeihong_bec_branching_all.txt};

\addplot[color=myParula01Blue,only marks,mark=x] table[x=erasure,y=branching] {figures/128_64_polar_bec_branching_all.txt};

\end{axis}
\end{tikzpicture}
	\vspace*{-3mm}
	\caption{$\E[G]$ vs. $\epsilon$ for the codes in Fig. \ref{fig:performance} with $n=128$ (solid lines: from Corollary \ref{cor:guessing}, markers: Monte-Carlo simulation averages).}\label{fig:complexity}
\end{figure}
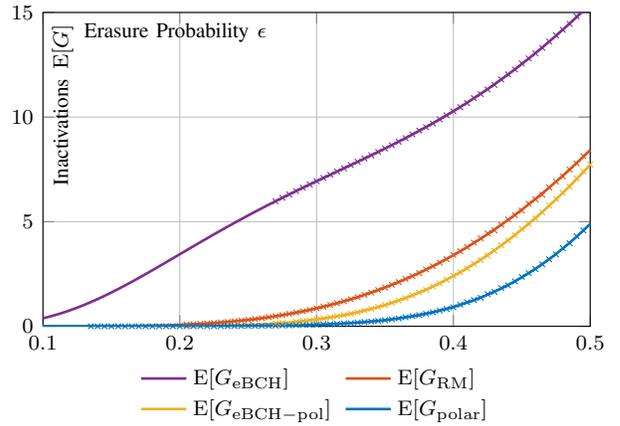
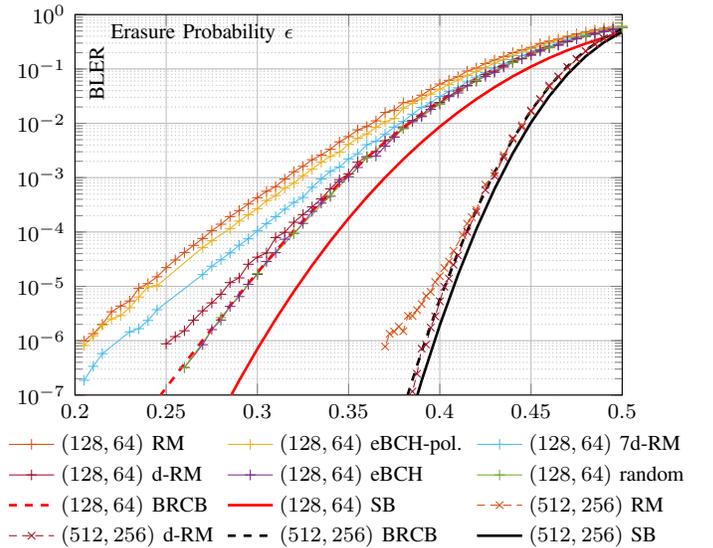
\begin{figure}
	\centering
	\begin{tikzpicture}
\begin{semilogyaxis}[
mark options={solid,scale=0.75},
width=1*\linewidth,
height=0.75*\columnwidth,
title style={font=\footnotesize,align=center},
legend cell align=left,
legend style={font=\footnotesize},
legend columns=3,
legend style={at={(0.5,-0.075)},anchor=north,draw=none,/tikz/every even column/.append style={column sep=1mm},cells={align=left}},
ylabel near ticks,
xlabel near ticks,
xmin=0.2,
xmax=0.5,
ymin=1e-7,
ymax=1,
xlabel={\textcolor{black}{Erasure Probability $\epsilon$}},
ylabel={\textcolor{black}{BLER}},
x label style={at={(axis description cs:0.0,0.9)},anchor=south west,yshift=0em,xshift=1em},
y label style={at={(axis description cs:0.01,0.75)},anchor=north west,yshift=0em,xshift=0em},
grid=both,
major grid style={solid,draw=gray!50},
minor grid style={densely dotted,draw=gray!50},
label style={font=\footnotesize},
tick label style={font=\footnotesize},
]

\addplot[myparula21] table[x=erasure,y=P] {figures/128_64_RM_bec.txt};
\addlegendentry{$(128,64)$ RM}
\addplot[myparula31] table[x=erasure,y=P] {figures/128_64_eBCHpeihong_bec.txt};
\addlegendentry{$(128,64)$ eBCH-pol.}
\addplot[myparula61] table[x=erasure,y=P] {figures/128_64_proposedRMdynamic_bec.txt};
\addlegendentry{$(128,64)$ 7d-RM}
\addplot[myparula71] table[x=erasure,y=P] {figures/128_64_RMsubcode_bec.txt};
\addlegendentry{$(128,64)$ d-RM}
\addplot[myparula41] table[x=erasure,y=P] {figures/128_64_ebch_bec.txt};
\addlegendentry{$(128,64)$ eBCH}
\addplot[myparula51] table[x=erasure,y=P] {figures/128_64_random_bec.txt};
\addlegendentry{$(128,64)$ random}
\addplot[color=red,line width=1pt,dashed] table[x=erasure,y=P] {figures/128_64_berlekamp_bec.txt};
\addlegendentry{$(128,64)$ BRCB}
\addplot[color=red,line width=1pt,solid] table[x=erasure,y=P] {figures/128_64_singleton_bec.txt};
\addlegendentry{$(128,64)$ SB}
\addplot[myparula22] table[x=erasure,y=P] {figures/512_256_RM_bec.txt};
\addlegendentry{$(512,256)$ RM}
\addplot[myparula72] table[x=erasure,y=P] {figures/512_256_RMsubcode_bec.txt};
\addlegendentry{$(512,256)$ d-RM}
\addplot[color=black,line width = 1pt,dashed] table[x=erasure,y=P] {figures/512_256_berlekamp_bec.txt};
\addlegendentry{$(512,256)$ BRCB}
\addplot[color=black,line width = 1pt,solid] table[x=erasure,y=P] {figures/512_256_singleton_bec.txt};
\addlegendentry{$(512,256)$ SB}

\end{semilogyaxis}
\end{tikzpicture}
	\vspace*{-8mm}
	\caption{BLER vs. $\epsilon$ for the codes with dynamic frozen bits (also RM codes).}\label{fig:performance_dynamic}
\end{figure}

Fig. \ref{fig:inac_evo_polar} shows the average number of unresolved inactivations $\E[G_i]$ as the \ac{SC} inactivation decoder with consolidations proceeds from $i=1$ to $i=n$ with $n=128$ at $\epsilon=0.4$ for the polar and \ac{RM} code. This highlights the performance vs. complexity trade-off. The code with a better performance, i.e., \ac{RM} code (see Fig. \ref{fig:performance}), has more inactivations at early stages due to more information bits at unreliable positions; hence a larger decoding complexity. The frozen bits placed at fairly reliable positions help resolve the inactivations, yielding a better performance. The decoder is not able to resolve the inactivated bits for the polar code because of the lack of frozen bits at reliable positions appearing after inactivations. Hence, $\E[G_i]$ provides a measure to quantify performance vs. complexity trade-off. On the one hand, many unresolved inactivations increase complexity and are, hence, undesired. On the other hand, too few inactivations to begin with do not make best use of the information provided by frozen bits when it comes resolving inactivations.

In Fig. \ref{fig:inac_evo}, $\E[G_i]$ is provided for codes with dynamic frozen bits. Observe the large number of inactivations for the random code, where they are mostly resolved at the end. The \ac{eBCH} and the d-\ac{RM} codes provide a similar performance (see Fig. \ref{fig:performance_dynamic}) with a lower complexity compared to the random code. In addition, the \ac{eBCH}-polar code has the lowest complexity; yet, with a degraded performance (see Fig. \ref{fig:performance_dynamic}). The 7d-RM code is an exemplary construction for a code performing halfway between the \ac{eBCH}-polar and d-\ac{RM} codes (see Fig. \ref{fig:performance_dynamic}). An analysis of the additional complexity due to dynamic frozen bits is left to future work.

\begin{figure}    
	\centering
	\begin{tikzpicture}
\begin{axis}[
mark options={solid,scale=0.75},
width=1*\linewidth,
height=0.5*\columnwidth,
title style={font=\footnotesize,align=center},
legend cell align=left,
legend style={font=\footnotesize},
legend columns=2,
legend style={at={(0.5,-0.125)},anchor=north,draw=none,/tikz/every even column/.append style={column sep=1mm},cells={align=left}},
ylabel near ticks,
xlabel near ticks,
xmin=1,
xmax=128,
ymin=0,
ymax=3,
xlabel={\textcolor{black}{Bit index $i$}},
ylabel={\textcolor{black}{\scriptsize{Unresol. Inactiv. $\E[G_i]$}}},
x label style={at={(axis description cs:0.02,0.8)},anchor=south west,yshift=0em,xshift=1em},
y label style={at={(axis description cs:0.01,0.01)},anchor=north west,yshift=0em,xshift=0em},
grid=both,
major grid style={solid,draw=gray!50},
minor grid style={densely dotted,draw=gray!50},
label style={font=\footnotesize},
tick label style={font=\footnotesize},
]

\addplot[color=myParula02Orange,line width = 1pt,dashed] table[x=index,y=expectedInac] {figures/128_64_RM_inac_evo_eps_04.txt};
\addlegendentry{$(128,64)$ RM}
\addplot[color=myParula01Blue,line width = 1pt,dashed] table[x=index,y=expectedInac] {figures/128_64_polar_inac_evo_eps_04.txt};
\addlegendentry{$(128,64)$ polar}

\end{axis}
\end{tikzpicture}
	\vspace*{-3mm}
	\caption{$\E[G_i]$ vs. $i$ at $\epsilon=0.4$.}\label{fig:inac_evo_polar}
\end{figure}
\begin{figure}    
	\centering
	\begin{tikzpicture}
\begin{axis}[
mark options={solid,scale=0.75},
width=1*\linewidth,
height=0.6*\columnwidth,
title style={font=\footnotesize,align=center},
legend cell align=left,
legend style={font=\footnotesize},
legend columns=2,
legend style={at={(0.5,-0.1)},anchor=north,draw=none,/tikz/every even column/.append style={column sep=1mm},cells={align=left}},
ylabel near ticks,
xlabel near ticks,
xmin=1,
xmax=128,
ymin=0,
ymax=40,
xlabel={\textcolor{black}{Bit index $i$}},
ylabel={\textcolor{black}{\scriptsize{Unresol. Inactiv. $\E[G_i]$}}},
x label style={at={(axis description cs:0.02,0.85)},anchor=south west,yshift=0em,xshift=1em},
y label style={at={(axis description cs:0.01,0.25)},anchor=north west,yshift=0em,xshift=0em},
grid=both,
major grid style={solid,draw=gray!50},
minor grid style={densely dotted,draw=gray!50},
label style={font=\footnotesize},
tick label style={font=\footnotesize},
]

\addplot[color=myParula05Green,line width = 1pt,dashed] table[x=index,y=expectedInac] {figures/128_64_random_inac_evo_eps_04.txt};
\addlegendentry{$(128,64)$ random}

\addplot[color=myParula04Purple,line width = 1pt,dashed] table[x=index,y=expectedInac] {figures/128_64_ebch_inac_evo_eps_04.txt};
\addlegendentry{$(128,64)$ eBCH}
\addplot[color=myParula07Red,line width = 1pt,dashed] table[x=index,y=expectedInac] {figures/128_64_RMsubcode_inac_evo_eps_04.txt};
\addlegendentry{$(128,64)$ d-RM}
\addplot[color=cyan,line width = 1pt,dotted] table[x=index,y=expectedInac] {figures/128_64_RMproposedLimited_inac_evo_eps_04.txt};
\addlegendentry{$(128,64)$ 7d-RM}
\addplot[color=myParula03Yellow,line width = 1pt,dashed] table[x=index,y=expectedInac] {figures/128_64_ebchpeihong_inac_evo_eps_04.txt};
\addlegendentry{$(128,64)$ eBCH-pol.}
\draw (75,-1) rectangle (101,1);
\draw (75,0) -- (44.05,11.9);
\draw (101,1) -- (86.3,31.8);
\coordinate (insetPosition) at (rel axis cs:0.7,0.2);
\end{axis}

\begin{axis}[at={(insetPosition)},anchor={outer south east},tiny,xmin=75,xmax=101,ymin=0,ymax=1]

\addplot[color=myParula07Red,line width = 1pt,dashed] table[x=index,y=expectedInac] {figures/128_64_RMsubcode_inac_evo_eps_04sub.txt};
\addplot[color=cyan,line width = 1pt,dotted] table[x=index,y=expectedInac] {figures/128_64_RMproposedLimited_inac_evo_eps_04sub.txt};
\addplot[color=myParula03Yellow,line width = 1pt,dashed] table[x=index,y=expectedInac] {figures/128_64_ebchpeihong_inac_evo_eps_04sub.txt};

\end{axis}

\end{tikzpicture}
	\vspace*{-3mm}
	\caption{$\E[G_i]$ vs. $i$ at $\epsilon=0.4$.}\label{fig:inac_evo}
\end{figure}
	\section{Conclusion}
An inactivation decoder is proposed using the schedule of the \ac{SC} decoding for transmission over the BEC. Using density evolution, the expected number of inactivations is derived analytically. The results are illustrated numerically for various codes. The expected number of unresolved inactivations quantifies the performance vs. complexity trade-off for the proposed SC inactivation or SCL decoder for \ac{MAP} decoding of a given code.
	
	\section*{Acknowledgement}
	This work was supported by the research grant "Efficient Coding and Modulation for Satellite Links with Severe Delay Constraints" funded by the Helmholtz Gemeinschaft through the HGF-Allianz DLR@Uni project Munich Aerospace.
	The authors thank Gianluigi Liva (DLR) for fruitful discussions that inspired this work, Peihong Yuan (TUM) for providing the dynamic frozen bit constraints for the \ac{eBCH} code and Gerhard Kramer (TUM) for comments improving the presentation.

	\appendix
\subsection{Proof of Proposition 1}
Note that the relation $P^{\mathrm{(SC)}}\leq P^{\mathrm{(GA)}}$ is shown already in \cite{arikan2009channel} while the equality is formalized later \cite{Mori:2009SIT,Sasoglu12}. Let $\mathcal{T} \triangleq \{(u_1^n, y_1^n)\in\mathcal{X}^n\times\mathcal{Y}^n:u_{\mathcal{A}^c} = 0^{n-k}\}$. For each sample point $(u_1^n, y_1^n)\in\mathcal{T}$, the \ac{SC} decoder outputs a vector $\hat{u}_1^n(u_1^n,y_1^n)$ whose elements are computed recursively as
\begin{equation}
\hat{u}_i(u_1^n,y_1^n) = \left\{\begin{array}{lll}
u_i &\text{if} \,\,\, i\in\mathcal{A}^c \\
f_i(y_1^n,\hat{u}_1^{i-1}(u_1^n,y_1^n))& \text{otherwise.}
\end{array}
\right.
\label{eq:decision_proof}
\end{equation}

The event of having the first bit-error at the $i$-th bit (information bit) under \ac{SC} decoding is defined as $\mathcal{B}_{i}^{\mathrm{(SC)}}\triangleq \{(u_1^n, y_1^n)\in\mathcal{T}:\hat{u}_1^{i-1}(u_1^n,y_1^n) = u_1^{i-1}, f_i(y_1^n,\hat{u}_1^{i-1}(u_1^n,y_1^n))\neq u_i\}$. However,
the bit-error event for the same bit $u_i$ under a genie-aided \ac{SC} decoder is $\mathcal{B}_{i}^{\mathrm{(GA)}} \triangleq \{(u_1^n, y_1^n)\in\mathcal{T}: f_i(y_1^n,u_1^{i-1}) \neq u_i\}$. Observe that $\mathcal{B}_{i}^{\mathrm{(SC)}}\subseteq\mathcal{B}_{i}^{\mathrm{(GA)}}$\cite{arikan2009channel}.
\begin{lemma}\label{lem:union}
	$\bigcup\limits_{i=1}^{\ell} \mathcal{B}_{i}^{\mathrm{(SC)}} = \bigcup\limits_{i=1}^{\ell} \mathcal{B}_{i}^{\mathrm{(GA)}}$ for all $\ell\in\{1,\dots,n\}$.
\end{lemma}
\begin{proof}
	We use induction. It is trivial for $\ell=1$. Assume that it holds for some $\ell = \ell'$.
	Note that the sets $\mathcal{B}_{i}^{\mathrm{(SC)}}$ are disjoint for $i=1,\dots,n$. Therefore, the proof is concluded if it is shown that $\mathcal{B}_{\ell'+1}^{\mathrm{(SC)}} = \mathcal{B}_{\ell'+1}^{\mathrm{(GA)}}\setminus\bigcup\limits_{i=1}^{\ell'}\mathcal{B}_{i}^{\mathrm{(GA)}}$. To this end, we write
	\vspace{-4mm}
	\begin{align}
		\mathcal{B}_{\ell'+1}^{\mathrm{(GA)}}\setminus\bigcup\limits_{i=1}^{\ell'}&\mathcal{B}_{\ell'+1}^{\mathrm{(GA)}} = \mathcal{B}_{\ell'+1}^{\mathrm{(GA)}}\setminus\bigcup\limits_{i=1}^{\ell'} \mathcal{B}_{i}^{\mathrm{(SC)}} \label{eq:assumption}\\ 
		& = \mathcal{B}_{\ell'+1}^{\mathrm{(GA)}}\setminus \{(u_1^n, y_1^n)\in\mathcal{T}: (f_1(y_1^n,u_1^0), \\
		&\qquad f_2(y_1^n,u_1^1),\dots,f_{\ell'}(y_1^n,u_1^{\ell'-1}))\neq u_1^{\ell'}\} \label{eq:step2}\\
		& = \mathcal{B}_{\ell'+1}^{\mathrm{(SC)}} \label{eq:final}
	\end{align}
	where \eqref{eq:assumption} follows from the induction hypothesis, \eqref{eq:step2} from the unions starting from $i=1$ to $i=\ell'$, e.g., the first union is
	\begin{equation}
		 \mathcal{B}_{1}^{\mathrm{(SC)}}\cup \mathcal{B}_{2}^{\mathrm{(SC)}} = \{(u_1^n, y_1^n)\in\mathcal{T}: (f_{1}(y_1^n,u_1^0),f_{2}(y_1^n,u_1^1))\neq u_1^{2}\}.
	\end{equation}
	Finally, we have \eqref{eq:final} by combining the definitions of $\mathcal{B}_{\ell'+1}^{\mathrm{(SC)}}$ and $\mathcal{B}_{\ell'+1}^{\mathrm{(GA)}}$.
\end{proof}
The result follows from Lemma \ref{lem:union} by setting $\ell = n$.\qed
	\IEEEtriggeratref{17}

\end{document}